\documentclass[a4paper,12pt]{article}

\usepackage[utf8]{inputenc} 
\usepackage[osf,noBBpl]{mathpazo} 
\usepackage{microtype}

\usepackage[a4paper]{geometry}
\usepackage{comment}

\usepackage{amsmath, amsthm, amsfonts, amssymb, amscd, natbib, bm}
\usepackage{graphicx, comment}
\usepackage{subfigure}
\usepackage{algorithm}
\usepackage{algorithmic}

\usepackage{xcolor} \definecolor{darkgreen}{rgb}{0,0.6,0} \definecolor{darkred}{rgb}{0.6,0,0}
\usepackage[pdfstartview=FitH,
            bookmarks=true,
            bookmarksnumbered=true,
            pdftex,
            pdfpagemode=UseOutlines,
            pdfpagelayout=OneColumn,
            colorlinks=true,
            urlcolor=blue,
            citecolor=darkgreen,
            pdfborder={0 0 0},
            backref=page]{hyperref}

\theoremstyle{plain}
\newtheorem{theorem}{Theorem}
\newtheorem{proposition}[theorem]{Proposition} 

\newtheorem{lemma}[theorem]{Lemma}

\numberwithin{theorem}{section}

\theoremstyle{definition}
\newtheorem{definition}[theorem]{Definition}

\newcommand{\abs}[1]{\lvert#1\rvert}
\newcommand{\norm}[1]{\|#1\|}
\newcommand{\inner}[1]{\langle#1\rangle}

\newcommand{\real}{\mathbb{R}}

\newcommand{\integer}{\mathbb{Z}}

\newcommand{\im}{\mathrm{Im}}

\newcommand{\CO}{\mathcal{O}}

\newcommand{\CL}{\mathcal{L}}
\newcommand{\myobot}{\mathop{\bigcirc\kern-15pt\perp}\nolimits}

\newcommand{\Ker}{\mathrm{ker}}

\newcommand{\mc}{\mbox{, }}

\newcommand{\Ext}{\mathrm{Ext}}

\title{Computing an LLL-reduced basis of the orthogonal lattice}

\author{
Jingwei Chen \\
Chongqing Key Lab of Automated Reasoning \& Cognition, \\
Chongqing Institute of Green and Intelligent Technology,\\ 
Chinese Academy of Sciences, China
\\\texttt{chenjingwei@cigit.ac.cn}
\and Damien Stehl\'e\\
ENS de Lyon, Laboratoire LIP \\
(UMR CNRS - ENS Lyon - UCB Lyon 1 - INRIA 5668),  France\\
\texttt{damien.stehle@ens-lyon.fr}
\and Gilles Villard\\
CNRS, Laboratoire LIP\\
(UMR CNRS - ENS Lyon - UCB Lyon 1 - INRIA 5668),  France\\
\texttt{gilles.villard@ens-lyon.fr}
}

\begin{document}

\date{\today}
\maketitle

\begin{abstract}
	As a typical application, the Lenstra-Lenstra-Lov\'asz  lattice basis reduction algorithm (LLL) is used to compute a reduced basis of the orthogonal lattice for a given integer matrix, via reducing a special kind of lattice bases. With such bases in input, we propose a new technique for bounding from above the number of iterations required by the LLL algorithm. The main technical
	ingredient is a variant of the classical LLL potential, which could prove useful
	to understand the behavior of LLL for other families of input bases.
\end{abstract}


\section{Introduction}

Let~$k<n$~be two positive integers. Given a full column rank $n\times k$ integer
matrix ${\bf A}=(a_{i,j})$, we study the behaviour of the Lenstra-Lenstra-Lov\'asz  algorithm~\cite{LenstraLenstraLovasz1982}
for computing a reduced basis for the \emph{orthogonal lattice} of ${\bf A}$
\begin{equation}
\label{eq:orthlatt}
\CL^\perp({\bf A}) = \left\{\bm m\in\integer^n: {\bf A}^T\bm m=\bm 0\right\}=\Ker({\bf A}^T)\cap\integer^n.
\end{equation}
The algorithm proceeds by unimodular column transformations from the input matrix
$\Ext_K({\bf A}) \in \integer^{(n+k) \times n}$:
\begin{equation}
\label{eq:CA}
\Ext_K({\bf A}) :=
\begin{pmatrix}
K\!\cdot\!{\bf A}^T \\
{\bf I}_n \\
\end{pmatrix} =\begin{pmatrix}
K\!\cdot\!a_{1,1} & K\!\cdot\!a_{2,1} & \cdots &K\!\cdot\!a_{n,1} \\
\vdots & \vdots &\ddots &\vdots\\
K\!\cdot\!a_{1,k} & K\!\cdot\!a_{2,k} & \cdots & K\!\cdot\!a_{n,k} \\
1 & 0 &\cdots  & 0 \\
0 & 1 & \cdots & 0 \\
\vdots & \vdots & \ddots & 0 \\
0 & 0 & \cdots & 1 \\
\end{pmatrix}.
\end{equation}
where $K$ is a sufficiently large positive integer. The related definitions and the LLL algorithm are given
in Section~\ref{sec:pre}.
The reader may refer to~\cite{NguyenVallee2010} for a comprehensive review of LLL, and to 
\cite{Schmidt1968} and \cite{NguyenStern1997} concerning the orthogonal lattice.

Usual techniques gives that LLL reduction requires  $\CO(n^2\log(K\cdot\norm{{\bf A}}))$ {\em swaps} (see Step~\ref{algostep:LLLswap} of Algorithm~\ref{algo:LLL}) for a  basis as in \eqref{eq:CA}, where $\norm{{\bf A}}$ bounds from above the  Euclidean norms of the rows and columns of~${\bf A}$.
We recall that most known  LLL reduction algorithms iteratively perform  two types of vector operations: translations and swaps.
The motivation for studying bounds on the number of swaps comes from the fact that
this number governs known cost analyses of the reduction.

Folklore applications of the reduction of bases as in~\eqref{eq:CA} include, for example, the computation
of integer relations between real numbers~\cite{HastadJustLagariasSchnorr1989,ChenStehleVillard2013}, the computation of minimal polynomials~\cite{KannanLenstraLovasz1984} (see also~\cite{NguyenVallee2010}). 
A main difficulty however, both theoretically and practically, remains to master the {\em scaling parameter}~$K$ that can be very large.  
Heuristic and practical 
solutions may for instance rely on a doubling strategy (successive trials with $K=2, 2^2, 2^4, \ldots$) for finding a suitable scaling. 
Or an appropriate value for $K$ may be derived from {\em a priori} bounds such as
heights of algebraic numbers~\cite{KannanLenstraLovasz1984} and may overestimate the smallest suitable
value for actual inputs. Since the usual bound on the number of swaps is linear in~$\log K$, the overestimation
could be a serious drawback. We  show that this may not be always the case.

We consider the reduction of a basis as in \eqref{eq:CA} for obtaining a basis of the orthogonal
lattice~(\ref{eq:orthlatt}). We establish a bound on the number of swaps that does not depend on $K$ as
soon as $K$ is above a threshold value (as specified in~(\ref{eq:chooseC})). 
This threshold depends only on the dimension and invariants of the orthogonal lattice.

\medskip
\noindent{\sc Our contribution.}
The  analyses of LLL and many LLL variants bound the number of iterations  using the geometric
decrease of a potential  that is defined using the
Gram-Schmidt norms of the basis vectors; see \eqref{eq:clsPot}.  We are going to see that this classical potential does not
capture a typical unbalancedness of the Gram-Schmidt norms that characterizes bases in~\eqref{eq:CA}.
Taking into account the latter structure will lead us to a better bound for the
number of iterations (see Table~\ref{tab:k}).
Intuitively, as the basis being manipulated becomes reduced, two groups of vectors are formed: some with small Gram-Schmidt norms, and some
others with large Gram-Schmidt norms. As soon they are formed, the two groups do not interfere much.

In Section~\ref{sec:potential} we introduce a new LLL potential function that generalizes the classical one for capturing the previously mentioned unbalancedness.
Its geometric decrease during the execution also leads to a bound on the number of iterations (see Theorem~\ref{th:pot}). In Section~\ref{sec:knapsack}, we specialize the potential to  the case of bases as in~\eqref{eq:CA} for computing the orthogonal lattice $\CL^\perp({\bf A})$. As discussed above, we will see that at some point the number of iterations can be shown to be independent of the scaling parameter~$K$, or, in other words, independent of a further increase of the input size. We note that this new potential is defined for all lattice bases, but it may not always lead to better 
bounds on the number of LLL iterations. 

\medskip
\noindent{\sc Related work.} 
The extended gcd algorithm in~\cite{HavasMajewskiMatthews1998}
uses a basis as in~\eqref{eq:CA} with $k=1$. It is shown in~\cite[Sec.\,3, p.\,127]{HavasMajewskiMatthews1998}
that if $K$ is sufficiently large, then the sequence of operations performed by  LLL is independent of~$K$.
A somewhat similar remark had been made in \cite{Pohst1987}. 
We also note that in the analysis of the gradual sub-lattice reduction algorithm 
of~\cite{vanHoeijNovocin2012}, a similar separation of large and small basis vectors was used, also for a better bound on the number of iterations. 
Our new potential function allows a better understanding of the phenomenon.

We see our potential function for LLL as a new complexity analysis tool 
that may help further theoretical and practical studies of LLL and its applications. 
Various approaches exist for computing the orthogonal lattice ${\bf A}$, or equivalently 
an integral kernel basis of ${\bf A}^T$. A detailed comparison of the methods remains to be done and would be 
however outside the scope of this paper that focuses on the properties of the potential. 
An integral kernel basis may be obtained from a unimodular multiplier for the Hermite 
normal form of~${\bf A}$~\cite{StorjohannLabahn1996} (see also~\cite{Storjohann2005} for the related linear system solution 
problem), which may be combined as in~\cite[Ch.\,8]{Sims1994} and~\cite{ChenStorjohann2005}
with LLL for minimizing the bit size of the output. 
A direct application of LLL to $\Ext_K({\bf A})$ is an important alternative solution.  We refer to~\cite{Stehle2017} and references therein 
concerning existing LLL variants.

\smallskip\noindent{\sc Future work.}
Future research directions are to apply this potential to bit complexity studies of the LLL
basis reduction~\cite{Storjohann1996, NovocinStehleVillard2011, NeumaierStehle2016}, especially for specific input bases. 
Indeed, an interesting problem is  to design an algorithm for computing a reduced basis for $\CL^\perp({\bf A})$ that features a bit complexity bound independent of the scaling parameter, and to compare it to approaches based on the Hermite normal form.

\smallskip\noindent{\sc Notations.} Throughout the paper, vectors are in column and denoted in bold. For $\bm x\in\real^m$,
$\norm{\bm x}$ is the Euclidean norm of~$\bm x$. 
Matrices are denoted by upper case letters in bold, such as ${\bf A}$, ${\bf B}$, etc. For a matrix ${\bf A}$, ${\bf A}^T$ is the transpose of ${\bf A}$, and $\norm{{\bf A}}$ bounds the Euclidean norms of the columns and rows of~${\bf A}$. 
The base of logarithm is~$2$.

\section{Preliminaries}
\label{sec:pre}
We give some basic definitions and results that are needed for the rest of the paper. A comprehensive
presentation of the LLL algorithm and its applications may be found in~\cite{NguyenVallee2010}.

\smallskip\noindent{\sc Gram-Schmidt orthogonalization.}
Let $\bm b_1, \cdots,\bm b_n\in\real^m$ be linearly independent vectors. Their \emph{Gram-Schmidt orthogonalization} $\bm b_1^*,\cdots,\bm b_n^*$ is defined as follows:
\[
\bm b_1^* = \bm b_1 \ \mbox{ and } \
\forall i>1: \bm b_i^* = \bm b_i - \sum_{j=1}^{i-1}\mu_{i,j}\bm b_j^*,
\]
where the $\mu_{i,j} = \frac{\inner{\bm b_i, \bm b_j^*}}{\inner{\bm b_j^*, \bm b_j^*}}$ for all~$i>j$ are called the \emph{Gram-Schmidt coefficients}. We call the~$\|\bm b_i^*\|$'s the \emph{Gram-Schmidt norms} of the~$\bm b_i$'s.

\medskip\noindent{\sc Lattices.} A \emph{lattice} $\Lambda\subseteq\real^m$ is a discrete additive subgroup of $\real^m$. If $(\bm b_i)_{i\le n}$ is a set of generators for $\Lambda$, then
\[
\Lambda =\CL(\bm b_1,\ldots,\bm b_n)=\left\{\sum_{i=1}^n z_i\bm b_i:
\, z_i\in\integer\right\}.
\]
If the $\bm b_i$'s are linearly independent, then they are said to form
a \emph{basis} of~$\Lambda$. When $n\ge 2$, there exist infinitely many bases for a lattice. Every basis is related by an integral unimodular transformation (a linear transformation with determinant $\pm 1$) to any other. Further, the number of vectors of different bases of a lattice~$\Lambda$ is always the same, and we call this number the \emph{dimension} of the lattice, denoted by $\dim(\Lambda)$. If ${\bf B}=(\bm b_1,\ldots,\bm b_n)\in\real^{m\times n}$ is a basis for a lattice $\Lambda=\CL({\bf B})$, the \emph{determinant} of the lattice is defined as $\det(\Lambda)= \sqrt{\det({\bf B}^T{\bf B})}$. It is invariant across all bases of~$\Lambda$.

\smallskip\noindent{\sc Successive minima.} For a given lattice $\Lambda$, we let $\lambda_1(\Lambda)$ denote the minimum Euclidean norm of vectors in $\Lambda \setminus \{\bm 0\}$. From  Minkowski's first theorem, we have  $\lambda_1(\Lambda)\le \sqrt{n}\cdot \det(\Lambda)^{1/n}$,
where~$n = \dim (\Lambda)$. More generally, for all $1\le i\le n$, we define the $i$-th \emph{minimum} as
\[\lambda_i(\Lambda)=\min_{\begin{array}{c}\bm v_1,\cdots,\bm v_i\in\Lambda\\\mbox{\,linearly independent}\end{array}}{\max_{j\le i}\norm{\bm v_j}}.\]
Minkowski's second theorem states that $\prod_{i\leq n}\lambda_i(\Lambda)\le\sqrt{n}^n \cdot \det(\Lambda)$.

\smallskip\noindent{\sc Sublattices.} Let $\Lambda\subseteq\real^n$ be a lattice. We say that $\Lambda^\prime$ is a \emph{sublattice} of $\Lambda$ if $\Lambda^\prime\subseteq\Lambda$ is a lattice as well. If $\Lambda^\prime$ is a sublattice of $\Lambda$ then $\lambda_i(\Lambda)\le \lambda_i(\Lambda^\prime)$ for $i\leq\dim(\Lambda^\prime)$. 
A sublattice $\Lambda^\prime$ of $\Lambda\subset\real^n$ is said to be \textit{primitive} if there exists a subspace $E$ of $\real^n$ such that $\Lambda^\prime=\Lambda\cap E$.

\smallskip\noindent{\sc Orthogonal lattices.} Given a full column rank matrix ${\bf A}\in\integer^{n\times k}$, the set $\CL^{\perp}({\bf A})$ defined in~\eqref{eq:orthlatt} forms a lattice, called the \emph{orthogonal lattice} of~${\bf A}$. We have  $\dim(\CL^{\perp}({\bf A}))= n- k$.
Using  
$\ker({\bf A}^T)^{\perp}=\im(A)$ and~\cite[Cor.\,p.\,328]{Schmidt1968} for primitive lattices   we have 
$$
\det(\CL^{\perp}({\bf A})) =\det(\integer^n\cap\ker({\bf A}^T)) =\det(\integer^n\cap\im({\bf A})),
$$
then $\CL({\bf A})\subseteq \integer^n
\cap\im({\bf A})$ and Hadamard's inequality lead to:
\begin{equation}
\label{eq:det-bound}
\det(\CL^{\perp}({\bf A})) \le\det(\CL({\bf A})) \le\norm{{\bf A}}^k.
\end{equation}

\smallskip\noindent{\sc LLL-reduced bases.}
The goal of lattice basis reduction is to find a basis with vectors as short and orthogonal to each other as possible. Among numerous lattice reduction notions, the LLL-reduction \cite{LenstraLenstraLovasz1982} is one of the most commonly used. Let $\frac{1}{4}<\delta< 1$. Let ${\bf B}=(\bm b_1, \ldots,\bm b_n)\in\real^{m\times n}$ be a basis of a lattice $\Lambda$. We say that ${\bf B}$ is \emph{size-reduced} if all Gram-Schmidt coefficients satisfy $\abs{\mu_{i j}}\leq \frac{1}{2}$. We say that ${\bf B}$ satisfies the \emph{Lov\'asz conditions} if for all~$i$ we have $\delta  \norm{\bm b_i^*}^2\le\norm{\bm b_{i+1}^*}^2 +\mu_{i+1, i}^2\norm{\bm b_i^*}^2$. If a basis ${\bf B}$ is size-reduced and satisfies the Lov\'asz conditions, then we say that ${\bf B}$ is \emph{LLL-reduced} (with respect to the parameter~$\delta$). If a basis ${\bf B}=(\bm b_1,\ldots,\bm b_n)$ of $\Lambda$ is LLL-reduced, then we have:
\[
\forall i < n \mc\norm{\bm b_i^*}^2\le \alpha\norm{\bm b_{i+1}^*}^2,
\]
\begin{equation}\label{eq:lll2}
\forall i\leq n \mc\norm{\bm b_i}^2\le \alpha^{i-1}\norm{\bm b_{i}^*}^2,
\end{equation}
\begin{equation}\label{eq:lll-gen}
\forall  i \leq j\leq n \mc\norm{\bm b_i}\leq \alpha^{\frac{n-1}{2}}\lambda_j(\Lambda),
\end{equation}
where~$\alpha =\frac{4}{4\delta -1}$.
In particular, we have $\norm{\bm b_1}\leq \alpha^{\frac{n-1}{2}}\lambda_1(\Lambda)$. In this paper, we use the original LLL
parameter $\delta = \frac{3}{4}$ and hence $\alpha = 2$.

\smallskip\noindent{\sc The LLL algorithm.} We now sketch the LLL algorithm. Although there exist many LLL variants in the literature, most of them follow the following structure. Step~\ref{algostep:LLLswap} is called an \emph{LLL swap}.

\begin{algorithm}
	\caption{(LLL)}
	\label{algo:LLL}
	\begin{algorithmic}[1]
		\REQUIRE A basis $(\bm b_i)_{i\le n}$ of a lattice $\varLambda\subseteq\integer^n$.
		\ENSURE An LLL-reduced basis of $\varLambda$.
		\STATE $i:=2$;
		\WHILE{$i\leq n$}
		\STATE\label{algostep:szrdc} Size-reduce $\bm
		b_{i}$ by $\bm b_{1},\cdots,\bm b_{i-1}$;
		\IF{Lov\'asz condition holds for $i$}
		\STATE Set $i:=i+1$;
		\ELSE
		\STATE\label{algostep:LLLswap} (LLL swap) Swap $\bm b_i$ and $\bm b_{i-1}$; set $i:=\max\{i-1,2\}$;
		\ENDIF
		\ENDWHILE
		\STATE Return $(\bm b_i)_{i\le n}$.
	\end{algorithmic}
\end{algorithm}

To clarify the structure of the algorithm, we omit some details in the above description, e.g., the update of Gram-Schmidt coefficients. From the sketch, we see that we can bound the running-time of LLL by the number of while loop iterations times the cost of each iteration. In fact, most cost bounds for LLL variants proceed via this simple argument. 
It was showed in~\cite{LenstraLenstraLovasz1982} that the number of LLL swaps is  $\CO(n^2\log\norm{{\bf B}})$. 
The following lemma plays a very important role in the analysis of LLL; see \cite{LenstraLenstraLovasz1982} for a proof.
\begin{lemma}
	\label{lem:montonicityGS}
	Let $\bf B$ and $\bf B'$ be bases after and before an LLL swap between
	$\bm b_i$ and $\bm b_{i+1}$. Then
	\begin{eqnarray*}
		\max\{\norm{\bm b_{i}^{\prime*}},\norm{\bm b_{i+1}^{\prime*}}\}
		& \leq &
		\max\{\norm{\bm b_{i}^*},\norm{\bm b_{i+1}^*}\},
		\\
		\min\{\norm{\bm b_{i}^{\prime*}},\norm{\bm b_{i+1}^{\prime*}}\}
		& \geq &
		\min\{\norm{\bm b_{i}^*},\norm{\bm b_{i+1}^*}\},
		\\
		\norm{\bm b_{i}^{*}}\cdot\norm{\bm b_{i+1}^{*}} & = & \norm{\bm b_{i}^{\prime*}}\cdot\norm{\bm b_{i+1}^{\prime*}}, \\
		\frac{\norm{\bm b_{i+1}^{\prime*}}}{\norm{\bm b_{i+1}^{*}}}
		=\frac{\norm{\bm b_{i}^{*}}}{\norm{\bm b_{i}^{\prime*}}}
		& \ge& \frac{2}{\sqrt{3}},
		\\
		\forall j \notin \{i,i+1\}&:& \bm b_{j}^{\prime*} =\bm b_{j}^{*}.
	\end{eqnarray*}
\end{lemma}


\section{A new potential}
\label{sec:potential}

In this section, we introduce a variant of the classical LLL potential~
\begin{equation}
\label{eq:clsPot}
\Pi({\bf B})= \sum_{i=1}^{n-1} (n-i) \log \|{\bm b}_i^*\|
\end{equation} of a lattice basis~${\bf B}$. 
The variant we introduce is well-suited for analyzing the number of LLL swaps for the case that both the input and output bases have $k$ large Gram-Schmidt norms and $n-k$ small Gram-Schmidt norms, for some~$k<n$. This is for example the case
for the input basis as  \eqref{eq:CA}; see Section \ref{subsec:inputoutput}. The new potential is aimed at accurately measuring the progress made during the LLL execution, for such unbalanced bases.

\begin{definition}
	Let $k \leq n \leq m$ be positive integers and ${\bf B} \in \real^{m \times n}$ be
	full column rank. We let~$s_1 < \ldots < s_{n-k}$ be the indices of the $n-k$ smallest Gram-Schmidt norms of~${\bf B}$
	(using the lexicographical in case there are several $(n-k)$-th smallest Gram-Schmidt
	norms), and set~$S = \{ s_i\}_{i \leq n-k}$. We let~$\ell_1 < \ldots < \ell_{k}$ be the indices of the other $k$ Gram-Schmidt norms, and set~$L= \{\ell_j\}_{j \leq k}$.
	The \emph{$k$-th LLL potential} of~${\bf B}$ is defined as:
	\[
	\Pi_k({\bf B}) = \sum_{j=1}^{k-1}(k-j)\log\norm{\bm b_{\ell_j}^{*}} - \sum_{i=1}^{n-k} i\log\norm{\bm b_{s_i}^{*}} + \sum_{i=1}^{n-k}s_{i}.
	\]
\end{definition}

Note that for~$k=n$, we recover the classical potential~$\Pi$.
The rationale behind~$\Pi_k$ is that in some cases we know that the output basis
is made of vectors of very unbalanced Gram-Schmidt norms. As this basis is reduced, this means the first vectors have a small Gram-Schmidt norm, while the last vectors have large Gram-Schmidt norms. During the execution of LLL, such short and
large vectors do not interfere much. This is an unusual phenomenon: most often, long vectors are made shorter and short vectors are made longer, so that they are all balanced at the end. But this can happen if the long vectors are rather orthogonal to the short ones. When this is the case, LLL actually runs faster than usual, because it merely ``sorts'' the short vectors and the long vectors, without making them interact to create shorter vectors. Of course, it can do more intense computations among the short vectors and among the long vectors.
Unbalancedness of Gram-Schmidt norms is not captured by the classical potential, but it is with~$\Pi_k$. In particular, the new potential~$\Pi_k$ allows to not ``pay''
for the output unbalancedness in the analysis of the number of LLL swaps.

Similarly
to the classical potential, the $k$-th LLL potential monotonically decreases with the number of LLL swaps. More precisely, we have the following

\begin{proposition}\label{prop:mon}
	Let~${\bf B}$ and~${\bf B}'$ be the current $n$-dimensional lattice bases before and after an LLL swap.
	Then for any~$k \leq n$, we have $\Pi_k({\bf B}) - \Pi_k({\bf B}')
	\geq \log (2/\sqrt{3})$.
\end{proposition}

\begin{proof}
	Recall that  $S$ and $L$ are  the index sets for the $n-k$ Gram-Schmidt norms and the other $k$ Gram-Schmidt norms for the lattice basis ${\bf B}$. We define $S'$ and $L'$ for ${\bf B}'$ similarly.
	
	Suppose that this LLL swap occurs between $\bm b_\kappa$ and $\bm b_{\kappa+1}$. Then we must be in one of the following four cases.
	
	\medskip
	\noindent
	Case~1: $\kappa \in S$ and $\kappa+1 \in S$.
	
	Let $i_0\leq n-k$ such that $\kappa = s_{i_0}$ and
	$\kappa + 1 = s_{i_0+1}$.
	From Lemma~\ref{lem:montonicityGS}, we
	have $S' =    S$ and $L' = L$, and hence $\kappa = s_{i_0}'$ and $\kappa + 1 = s_{i_0+1}'$. For the other indices, we have $s_i' = s_i$ (for $i \leq n-k$) and $\ell_j' = \ell_j$ (for $j \leq k$). Then
	\begin{eqnarray*}
		\Pi_k({\bf B} ) - \Pi_k({\bf B}')  &=& \sum\limits_{j=1}^k(k-j)\log\frac{\norm{\bm b_{\ell_j}^{*}}}{\norm{\bm b_{\ell_j'}'}} + \sum\limits_{i=1}^{n-k}i\log\frac{\norm{\bm b_{s_i'}^{\prime*}}}{\norm{\bm b_{s_i}^{*}}} \\
		&& \text{\hspace*{0.4cm}}+ \sum\limits_{i=1}^{n-k}\left(s_{i} - s_{i}'\right)\\
		&=& i_0\log\frac{\norm{\bm b_{s_{i_0}'}^{\prime*}}}{\norm{\bm b_{s_{i_0}}^{*}}} + (i_0+1)\log\frac{\norm{\bm b_{s_{i_0+1}'}^{\prime*}}}{\norm{\bm b_{s_{i_0+1}}^{*}}}
		\\
		&=& \log\frac{\norm{\bm b_{\kappa+1}^{\prime*}}}{\norm{\bm b_{\kappa+1}^{*}}}\ge \log\left(\frac{2}{\sqrt{3}}\right),
	\end{eqnarray*}
	where the last inequality follows from Lemma~\ref{lem:montonicityGS}.
	
	\medskip
	\noindent
	Case 2: $\kappa \in L$ and $\kappa + 1 \in L$.
	
	The treatment of Case~1 can be adapted readily.

	\medskip
	\noindent
	Case 3: $\kappa\in L$, $\kappa+1\in S$, $S' = S$ and $L' = L$.
	
	Let $j_0\leq k$ such that $\kappa = \ell_{j_0}$, and $i_0\leq n-k$ such that $\kappa+1= s_{i_0}$. Then we have $\kappa = \ell_{j_0}'$
	and $\kappa + 1 = s_{i_0}'$. For the other indices,
	we have $s_i^{\prime} = s_i^{(t)}$ (for $i \leq n-k$) and $\ell_j^{\prime} = \ell_j^{(t)}$ (for $j \leq k$). Thus
	\begin{eqnarray*}
		\Pi_k({\bf B}) - \Pi_k({\bf B}^{\prime})  &=& \sum\limits_{j=1}^k(k-j)\log\frac{\norm{\bm b_{\ell_j}^{*}}}{\norm{\bm b_{\ell_j'}^{\prime*}}} + \sum\limits_{i=1}^{n-k}i\log\frac{\norm{\bm b_{s_i'}^{\prime*}}}{\norm{\bm b_{s_i}^{*}}} \\
		&& \text{\hspace*{0.4cm}}+ \sum\limits_{i=1}^{n-k}\left(s_{i} - s_{i}^{\prime}\right)\\
		&=& (k-j_0)\log\frac{\norm{\bm b_{\ell_{j_0}}^{*}}}{\norm{\bm b_{\ell_{j_0}'}^{\prime*}}} + i_0\log\frac{\norm{\bm b_{s_{i_0}'}^{\prime*}}}{\norm{\bm b_{s_{i_0}}^{*}}}\\
		&=& (k-j_0+i_0)\log\frac{\norm{\bm b_{\kappa+1}^{\prime*}}}{\norm{\bm b_{\kappa+1}^{*}}}\ge \log\left(\frac{2}{\sqrt{3}}\right),
	\end{eqnarray*}
	where the last inequality follows from Lemma~\ref{lem:montonicityGS} and the fact that~$k-j_0+i_0\geq 1$.

	\medskip
	\noindent
	Case 4: $\kappa\in L$, $\kappa+1\in S$, $S^{\prime} = S\cup\{\kappa\}\setminus\{\kappa+1\}$ and $L^{\prime} = L\cup\{\kappa+1\}\setminus\{\kappa\}$.
	
	Let $j_0\leq k$ such that $\kappa = \ell_{j_0}$, and $i_0\leq n-k$ such that $\kappa+1= s_{i_0}$.
	Then $\kappa = s_{i_0}^{\prime}$ and $\kappa + 1 = \ell_{j_0}^{\prime}$. For other indices, we have $s_i^{\prime} = s_i$ (for $i \leq n-k$) and $\ell_j^{\prime} = \ell_j$ (for $j \leq k$). Then
	\begin{eqnarray*}
		\Pi_k({\bf B}) - \Pi_k({\bf B}^{\prime})  &=& \sum\limits_{j=1}^k(k-j)\log\frac{\norm{\bm b_{\ell_j}^{*}}}{\norm{\bm b_{\ell_j'}^{\prime*}}} + \sum\limits_{i=1}^{n-k}i\log\frac{\norm{\bm b_{s_i'}^{\prime*}}}{\norm{\bm b_{s_i}^{*}}} \\
		&& \text{\hspace*{0.4cm}}+ \sum\limits_{i=1}^{n-k}\left(s_{i} - s_{i}^{\prime}\right)\\
		& = &(k-j_0)\log\frac{\norm{\bm b_{\ell_{j_0}}^{*}}}{\norm{\bm b_{\ell_{j_0}'}^{\prime*}}} + i_0\log\frac{\norm{\bm b_{s_{i_0}'}^{\prime*}}}{\norm{\bm b_{s_{i_0}}^{*}}} + 1 \\
		&= &(k-j_0)\log\frac{\norm{\bm b_{\kappa}^{*}}}{\norm{\bm b_{\kappa+1}^{\prime*}}}
		+ i_0 \log\frac{\norm{\bm b_{\kappa}^{\prime*}}}{\norm{\bm b_{\kappa+1}^{*}}} + 1 \\
		&\ge& 1,
	\end{eqnarray*}
	where the last inequality follows from Lemma~\ref{lem:montonicityGS}.
	The observation that~$1 \geq \log(2/\sqrt{3})$ allows to complete the proof.
\end{proof}

With the above property of the $k$-th LLL potential, we can bound the number of LLL swaps that LLL performs.

\begin{theorem}
	\label{th:pot}
	Let ${\bf B}\in\real^{m\times n}$ be a full column rank matrix.
	Let~${\bf B}'$ be the basis returned by the LLL algorithm when given~${\bf B}$ as input.  Then the number of swaps that LLL performs is no greater than
	\[
	\min_{1\le k\le n}\frac{\Pi_k({\bf B}) - \Pi_k({\bf B}')}{\log\left(\frac{2}{\sqrt{3}}\right)}.
	\]
\end{theorem}


\section{Orthogonal lattices}
\label{sec:knapsack}
As an application of the $k$-th LLL potential $\Pi_k$, we consider the problem of computing an LLL-reduced basis of an orthogonal lattice. Let~${\bf A} \in \integer^{n \times k}$ with~$n \geq k$. We aim at computing an LLL-reduced basis
of the orthogonal lattice~$\mathcal{L}^{\perp}({\bf A})$, by LLL-reducing~$\Ext_K({\bf A})$ (as defined in~\eqref{eq:CA}), for a sufficiently large integer~$K$.

In Subsection~\ref{subsec:correct}, we provide a sufficient
condition on the scaling parameter~$K$ so that a LLL-reduced basis of~$\mathcal{L}^{\perp}({\bf A})$ can be extracted from a LLL-reduced basis of~$\mathcal{L}(\Ext_K({\bf A}))$. For such a sufficiently large $K$, we study the Gram-Schmidt orthogonalizations of the input and output bases of the LLL call to~$\Ext_K({\bf A})$ in Subsection~\ref{subsec:inputoutput}, and  we provide a bound on the number of required
LLL swaps which is independent of $K$ in Subsection~\ref{subsec:terminate}.

\subsection{Correctness}
\label{subsec:correct}

For~$n \geq k$, we define $\sigma_{n,k}$ as the map that embeds $\real^n$ into $\real^{n+k}$ by adding $0$'s in the first $k$ coordinates.
\[
\begin{array}{rcl}
\sigma_{n,k}:\,\,\real^n&\rightarrow&\real^{n+k}\\
(x_1,\cdots,x_n)^T&\mapsto& (\underbrace {0,\cdots ,0 }_k, \underbrace{x_1,\cdots,x_n}_n)^T.
\end{array}
\]
We also define $\delta_{n,k}$ as the map that erases the first $k$ coordinates
of a vector in~$\real^{n+k}$.
\[
\begin{array}{rcl}	
\delta_{n,k}:\,\,\real^{n+k}&\rightarrow&\real^{n}\\
(x_1,\cdots, x_k,x_{k+1},\cdots,x_{k+n})^T&\mapsto& (x_{k+1},\cdots,x_{k+n})^T.
\end{array}
\]
We extend these functions to matrices in the canonical way.
The following proposition is adapted from~\cite[Theorem~4]{NguyenStern1997} (see also~\cite[Proposition~2.24]{Nguyen1999}). 
It shows that if~$K$ is sufficiently large, then calling the LLL algorithm on $\Ext_K({\bf A})$ provides an LLL-reduced basis
of~$\mathcal{L}^{\perp}({\bf A})$.

\begin{proposition}\label{prop:cor}
	Let ${\bf A}\in\integer^{n\times k}$ be full column rank and ${\bf B} = \Ext_K({\bf A})$.
	If~${\bf B}'$ is  an LLL-reduced basis of~$\mathcal{L}({\bf B})$ and
	\begin{equation}
	\label{eq:chooseC}
	K>2^{\frac{n-1}{2}} \cdot \lambda_{n-k}(\CL^\perp({\bf A})),
	\end{equation}
	then  $\delta_{n,k}({\bm b}_1'), \cdots, \delta_{n,k}({\bm b}_{n-k}')$ is an LLL-reduced basis of
	$\mathcal{L}^{\perp}({\bf A})$.
\end{proposition}

\begin{proof}
	As ${\bf A}\in\integer^{n\times k}$ is full column rank, we have
	$\dim(\mathcal{L}^{\perp}({\bf A}))=n-k$. For any basis ${\bf C} \in \integer^{n \times (n-k)}$ of $\mathcal{L}^{\perp}({\bf A})$, we have
	$\sigma_{n,k}({\bf C})={\bf B} \cdot {\bf C}$, and hence the lattice
	$\sigma_{n,k}(\mathcal{L}^\perp({\bf A}))$ is a sublattice of $\mathcal{L}({\bf B})$. This implies that, for all~$i \leq n-k$,
	\[
	\lambda_i(\mathcal{L}({\bf B})) \leq
	\lambda_i(\sigma_{n,k}(\mathcal{L}^\perp({\bf A})))
	= \lambda_i(\mathcal{L}^\perp({\bf A})).
	\]
	It follows from~\eqref{eq:lll-gen} that, for all~$i \leq n-k$,
	\begin{equation}\label{eq:C-cor}
	\norm{\bm b_i'}^2\le 2^{{n-1}} \cdot \lambda_{n-k}^2(\mathcal{L}({\bf B}))\le 2^{{n-1}} \cdot \lambda_{n-k}^2(\mathcal{L}^{\perp}({\bf A})).
	\end{equation}
	
	We now assume (by contradiction) that $\delta_{n,k}(\bm b_{i}')\notin\mathcal{L}^{\perp}({\bf A})$ for some $i \leq n-k$.
	Note that
	\[\bm b_{i}' = {\bf B} \cdot \delta_{n,k}(\bm b_{i}') = (K\cdot  \delta_{n,k}(\bm b_{i'})^T\cdot {\bf A}\,|\,\delta_{n,k}(\bm b_{i}')^T)^T.\]
	As the subvector $K\cdot \delta_{n,k}(\bm b_{i}')^T\cdot {\bf A}$ is non-zero, and 
	using the assumption on~$K$, we obtain that
	\[
	\begin{split}
	\norm{\bm b_{i}'}^2 & = \norm{K\cdot \delta_{n,k}(\bm b_{i}')^T\cdot {\bf A}}^2 + \norm{\delta_{n,k}(\bm b_{i}')}^2\\
	&\ge K^2  > 2^{{n-1}} \cdot \lambda_{n-k}^2(\mathcal{L}^{\perp}({\bf A})),
	\end{split}
	\]
	which contradicts~\eqref{eq:C-cor}.
	
	From the above, we obtain that $\delta_{n,k}({\bm b}_1'), \cdots, \delta_{n,k}({\bm b}_{n-k}')$ are linearly independent vectors in~$\mathcal{L}^{\perp}({\bf A})$. They actually form a basis of~$\mathcal{L}^{\perp}({\bf A})$. To see this, consider an arbitrary vector~${\bm c} \in  \mathcal{L}^{\perp}({\bf A})$. The vector~${\bf B} \cdot {\bm c}$ belongs to the
	real span of~${\bm b}_1', \cdots, {\bm b}_{n-k}'$ and
	to~$\mathcal{L}({\bf B})$. As~${\bf B}'$ is a basis of~$\mathcal{L}({\bf B})$, vector~${\bf B} \cdot {\bm c}$ is an integer combination of~${\bm b}_1', \cdots, {\bm b}_{n-k}'$ and vector~${\bm c}$ is an integer combination of~$\delta_{n,k}({\bm b}_1'), \cdots, \delta_{n,k}({\bm b}_{n-k}')$.
	
	Since ${\bf B}'$ is LLL-reduced and the first $k$ coordinates of each ${\bm b}_i'$ ($i\le n-k$) are~$0$, we obtain that $\delta_{n,k}({\bm b}_1'), \cdots, \delta_{n,k}({\bm b}_{n-k}')$ form an LLL-reduced basis of~$\mathcal{L}^\perp({\bf A})$.
\end{proof}

To make this condition on $K$ effective, we use some upper bounds on $\lambda_{n-k}(\mathcal{L}^{\perp}({\bf A}))$. For instance, from Minkowski's second theorem, we have
\[
\lambda_{n-k}(\mathcal{L}^{\perp}({\bf A}))\le (n-k)^{\frac{n-k}{2}} \cdot \det(\mathcal{L}^{\perp}({\bf A}))\le(n-k)^{\frac{n-k}{2}} \cdot\norm{{\bf A}}^k.
\]
Hence
\begin{equation}
\label{eq:effectiveK}
K>2^{\frac{n-1}{2}}\cdot(n-k)^{\frac{n-k}{2}} \cdot\norm{{\bf A}}^k
\end{equation}
suffices to guarantee that \eqref{eq:chooseC} holds.

The bound
in~\eqref{eq:effectiveK} can be very loose. Indeed, in many cases, we expect the minima of~$\mathcal{L}^{\perp}({\bf A})$ to be balanced, and if they are so, then
the following bound would suffice
\begin{equation}
\label{eq:effectiveK2}
K>2^{\Omega(n)} \cdot \norm{{\bf A}}^{\frac{k}{n-k}}.
\end{equation}

For such a scaling paramter $K$, according to Proposition \ref{prop:cor}, after termination of the LLL call with $\Ext_K(A)$ as its input, the output matrix must be of the following form:
\begin{equation}\label{eq:term}
\begin{pmatrix}
\bm 0 & {\bf M} \\
{\bf C} & {\bf N} \\
\end{pmatrix}\mc
\end{equation}
where the columns of ${\bf C}\in\integer^{n\times(n-k)}$ form an LLL-reduced basis of the lattice $\CL^\perp({\bf A})$. 
\footnote{In fact, the resulting matrix gives more information than an LLL-reduced basis of  $\CL^\perp({\bf A})$. For  instance, the columns of $\frac{1}{K}\cdot \mathbf{M}$ form a basis of the lattice generated by the rows of $\mathbf{A}$.}

\subsection{On the LLL input and output bases}
\label{subsec:inputoutput}

To bound the number of LLL swaps, we first investigate the matrix ${\bf B} = \Ext_K({\bf A})$ given as input to the LLL algorithm, and the output matrix ${\bf B}'$.

Intuitively, from the shape of ${\bf B}$ and the fact that ${\bf A}$ is full rank, there must be $k$ Gram-Schmidt norms of ${\bf B}$ that are ``impacted'' by the scaling parameter~$K$, and hence have large magnitude, while other $n-k$ Gram-Schmidt norms of ${\bf B}$ should be of small magnitude.

On the other hand, recall that $\mathbf{B}'$ is of the form \eqref{eq:term}. Since only the first $k$ coordinates are related to the scaling parameter~$K$, the submatrix ${\bf C}$ is ``independent'' of~$K$. Thus, each of $\norm{\bm b_1'^{*}}\mc\cdots\mc\norm{\bm b_{n-k}'^{*}}$ should be relatively small (for a sufficiently large~$K$), while each of $\norm{\bm b_{n-k+1}'^{*}}\mc\cdots\mc\norm{\bm b_n'^{*}}$ is ``impacted'' by~$K$, and hence with large magnitude.
The following result formalizes this discussion.

\begin{proposition}\label{prop:SG}
	Let ${\bf A}\in\integer^{n\times k}$ be of full column rank and ${\bf B}'$ the output basis of LLL with ${\bf B} = \Ext_K({\bf A})$ as input. If the scaling parameter $K \in \integer$ satisfies~\eqref{eq:chooseC}, then for the output matrix~${\bf B}'$ we have
	\[
	\begin{array}{lll}
	\forall i\leq n-k, & \forall j>n-k,& \ \|{\bm b}_i'^*\| < \|{\bm b}_j'^*\|.
	\end{array}
	\]

\end{proposition}

\begin{proof}
	From Proposition \ref{prop:cor}, we know that ${\bf B}'$ is of the form
	\[
	\begin{pmatrix}
	\bm 0 & {\bf *} \\
	{\bf C} & {\bf *} \\
	\end{pmatrix}\mc
	\]
	and that the columns of ${\bf C} \in \integer^{n \times k}$ form an LLL-reduced basis
	of $\CL^\perp({\bf A})$. We thus have, for $i\le n-k$
	\[
	\norm{\bm b_i'^{*}}^2 \le \norm{\bm b_i'}^2 = \norm{\bm c_i}^2
	\le 2^{n-k-1} \lambda_{n-k}^2(\CL^\perp({\bf A})).
	\]
	Further, for $n-k< j\le n$, we have
	\[
	\norm{\bm b_j'^{*}}^2\ge 2^{-k}\norm{\bm b_{n-k+1}'^{*}}^2\ge 2^{-k}K^2.
	\]
	The choice of~$K$ allows to complete the proof.
\end{proof}

We observe again that combining the condition of Proposition~\ref{prop:SG} together with
a general purpose bound on~$\lambda_{n-k}(\mathcal{L}^{\perp}({\bf A}))$ allows
to obtain a sufficient bound on~$K$ that can be efficiently derived from~${\bf A}$.

Although $\norm{\bm b_{s_i}^{*}}$ is relatively small with respect to~$K$, it can be bounded from below. In fact, we have a more general lower bound:
\begin{equation}
\label{eq:bndBs}
\forall i \leq n, \ \norm{\bm b_{i}^{*}}\ge 1.
\end{equation}
This is because that there is a coefficient in $\bm b_{i}$  which is equal to~$1$ and~$0$ for all other~${\bf b}_j$'s. This lower bound will be helpful in the proof of Theorem~\ref{thm:main}.

\subsection{Bounding the number of LLL swaps}
\label{subsec:terminate}

Suppose that $K$ is a sufficient large positive integer satisfying~\eqref{eq:chooseC}. Proposition~\ref{prop:cor} guarantees that we can use LLL with ${\bf B}=\Ext_K({\bf A})$ as input to compute an LLL-reduced basis for~$\mathcal{L}^{\perp}({\bf A})$.
We now study the number of LLL swaps performed in this call to the LLL algorithm.

\begin{theorem}
	\label{thm:main}
	Let ${\bf A}\in\integer^{n\times k}$ with a non-zero $k$-th principal minor,
	and $K$ an integer satisfying~\eqref{eq:chooseC}. Then,
	given ${\bf B}=\Ext_K({\bf A})$ as its input, LLL computes (as a submatrix of the returned basis) an
	LLL-reduced basis of $\CL^\perp({\bf A})$ after at most $\CO(k^3 + k(n-k)(1+\log\norm{{\bf A}}))$ LLL swaps, where $\norm{{\bf A}}$ is the maximum of the Euclidean norm of all rows and columns of the matrix ${\bf A}$.
\end{theorem}

\begin{proof}
	From Proposition \ref{prop:cor}, the LLL algorithm allows to obtain a LLL-reduced basis for $\CL^\perp({\bf A})$. We know from Theorem~\ref{th:pot} that in order to obtain an upper bound on the number of LLL swaps, it suffices to find an upper bound to~$\Pi_k ({\bf B})$ and a lower bound on~$\Pi_k({\bf B}')$, where~${\bf B}'$ is the basis returned by LLL when given~${\bf B}$ as input.
	From~\eqref{eq:bndBs} we have
	\[
	\begin{split}
	\Pi_k({\bf B}) &= \sum_{j=1}^k(k-j)\log\norm{\bm b_{\ell_j}^{*}} - \sum_{i=1}^{n-k}i\log\norm{\bm b_{s_i}^{*}} + \sum_{i=1}^{n-k}s_{i}\\
	& \leq  \sum_{j=1}^k(k-j)\log\norm{\bm b_{\ell_j}^{*}} + \sum_{i=1}^{n-k}s_{i} \\
	& \leq \sum_{j=1}^k(k-j)\log\norm{\bm b_{\ell_j}} + \sum_{i=1}^{n-k} (k+i)\\
	&\leq(1+\log K +\log\norm{{\bf A}}) \frac{k(k-1)}{2} + \frac{(n-k)(n+k+1)}{2}.
	\end{split}
	\]
	Thanks to  Proposition~\ref{prop:SG}, we have
	\[
	\begin{split}
	\Pi_k({\bf B}') &= \sum_{j=1}^k(k-j)\log\norm{\bm b_{\ell_j'}^{\prime*}} - \sum_{i=1}^{n-k}i\log\norm{\bm b_{s_i'}'^{*}} + \sum_{i=1}^{n-k}s_{i}'\\
	&= \sum_{j=1}^k(k-j)\log\norm{\bm b_{n-k+j}'^{*}} - \sum_{i=1}^{n-k}i\log\norm{\bm b_{i}'^{*}} + \sum_{i=1}^{n-k}i.
	\end{split}
	\]
	Since the first $k$ coefficients of $\bm b_{i}'^{*}$ are $0$ (for~$i\leq n-k$) and~${\bf A}$ is full-rank,  we must have $\norm{\bm b_{n-k+1}'^{*}}\ge K$. Further,
	since ${\bf B}'$ is LLL-reduced, combining with~\eqref{eq:lll2} we have, for $j \leq k$
	\[
	\norm{\bm b_{n-k+j}'^{*}}\ge 2^{\frac{1-j}{2}}\norm{\bm b_{n-k+1}'^{*}}\ge 2^{\frac{1-j}{2}}K\ge 2^{\frac{1-k}{2}}K.
	\]
	We hence obtain
	\[
	\begin{split}
	\Pi_k({\bf B}')
	&\ge\left(\log K +\frac{1-k}{2}\right) \sum_{j=1}^{k} (k-j) -
	\sum_{i=1}^{n-k} i \log\norm{\bm b_{i}'^*}  \\
	&\text{\hspace*{0.4cm}}+ \frac{(n-k)(n-k+1)}{2} \\
	&\ge
	\frac{k(k-1)}{2}\left(\log K +\frac{1-k}{2}\right) - (n-k)\sum_{i=1}^{n-k}
	\log\norm{\bm b_{i}'^*}\\
	&\text{\hspace*{0.4cm}}+ \frac{(n-k)(n-k+1)}{2},
	\end{split}
	\]
	where we used the fact that all~$\|{\bf b}_i'^*\|$'s are~$\geq 1$. 
	This is true
	for the~$\|{\bf  b}_i^*\|$'s and LLL cannot make the minimum Gram-Schmidt
	norm decrease. Using~\eqref{eq:det-bound}, we obtain:
	\[
	\begin{split}
	\Pi_k({\bf B}')&
	\ge \frac{k(k-1)}{2}\left(\log K +\frac{1-k}{2}\right) - (n-k)k \log \|{\bf A}\| \\
	&\text{\hspace*{0.4cm}}+ \frac{(n-k)(n-k+1)}{2}.
	\end{split}
	\]
	
	Finally, using Theorem~\ref{th:pot}, we obtain that the number of LLL swaps is no greater than
	\[
	\frac{\Pi_k({\bf B}) - \Pi_k({\bf B}')}{\log\left(\frac{2}{\sqrt{3}}\right)}
	\le\frac{k(n-\frac{k}{2})\log\norm{{\bf A}} + k^3 + (n-k)k}{\log\left(\frac{2}{\sqrt{3}}\right) },
	\]
	which is of $\CO(k^3+k(n-k)(1+\log\norm{{\bf A}}))$.
\end{proof}

In Table~\ref{tab:k} we compare favorably ($k=1,n/2$) the result of Theorem~\ref{thm:main} to the bounds on the number of swaps 
using the classical potential~(\ref{eq:clsPot}) and~$K$ fixed 
from the general threshold~\eqref{eq:effectiveK} or the heuristic one~\eqref{eq:effectiveK2}.  
We also consider $k=n-1$. However, in the latter case the problem reduces to linear 
system solving, and different techniques such as those in~\cite{Storjohann2005} should be considered.

\begin{table}[H]
	\caption{Upper bounds on the number of LLL swaps for different $k$ ($K$ sufficiently large), $\alpha = \log\norm{{\bf A}}$.}
	\label{tab:k}
	\begin{tabular}{lccc}
		\hline 
		& Classical analysis \eqref{eq:effectiveK} & Heuristic \eqref{eq:effectiveK2}& 
		New analysis   \\
		\hline 
		$k=1$&  $\CO(n^2\log n + n\alpha)$ &  $\CO(n^2 + n\alpha)$ &   $\CO(n\alpha)$ \\
		$k={n}/{2}$&  $\CO(n^3\log n + n^3\alpha)$  &  $\CO(n^3 + n^2\alpha)$ & $\CO(n^3+ n^2\alpha)$  \\
		$k=n-1$& $\CO(n^2\alpha)$ &   $\CO(n^2\alpha)$ &  $\CO(n^3+ n\alpha)$\\
		\hline 
	\end{tabular}
\end{table}

With the potential function $\Pi$ of~\eqref{eq:clsPot}, 
we have
\[
\begin{split}
\Pi({\bf B})
&\leq \log\prod_{i \leq n} \left(K^2 \|{\bf A}\|^2\right)^{\frac{\min(k,i)}{2}}\\ &\le
\frac{k(2n-k+1)}{2}\log\left(K \|{\bf A}\|\right).
\end{split}
\]
The bound on the number of LLL swaps obtained using the classical potential
is therefore~$\CO(k(n-k/2) (1+\log K + \log \|{\bf A}\|)$. While we see from Theorem \ref{thm:main} that  
the actual number of swaps for computing an LLL-reduced basis for $\CL^\perp({\bf A})$ does not grow with~$K$ when~$K$ is sufficiently large.

\section*{Acknowledgments}
	Our thanks go to anonymous 	referees for helpful comments, which make the presentation of the paper better. Jingwei Chen was partially supported by NNSFC (11501540, 11671377, 11771421) and Youth Innovation Promotion Association, CAS. Damien Stehl\'e was supported by ERC Starting Grant ERC-2013-StG-335086-LATTAC.



\newcommand{\noopsort}[1]{}


\end{document}